\documentclass[sigplan,screen]{acmart}

\usepackage{amsmath,amsfonts}
\usepackage{algorithmic}
\usepackage{graphicx}
\usepackage{textcomp}
\usepackage{xcolor}
\usepackage{stmaryrd}
\usepackage{listings}
\usepackage{tikz}
\usepackage{pgfplots} 
\pgfplotsset{width=10cm,compat=1.9}
\usepgfplotslibrary{statistics}
\pgfplotsset{compat=newest}
\usetikzlibrary{shapes,arrows,positioning}
\usetikzlibrary{shapes.arrows,patterns}
\usepackage{arydshln}
\usepackage{tabularx}
\usepackage{colortbl}
\usepackage{multirow}
\usepackage{amsthm}
\usepackage{microtype}

\def\BibTeX{{\rm B\kern-.05em{\sc i\kern-.025em b}\kern-.08em
    T\kern-.1667em\lower.7ex\hbox{E}\kern-.125emX}}

\newcommand{\mycc}{\cellcolor{lightgray}}
\theoremstyle{definition}
\newtheorem{theorem}{Theorem}
\newtheorem{definition}{Definition}
\newtheorem{lemma}[theorem]{Lemma}

\usepackage{soul}

\AtBeginDocument{%
  \providecommand\BibTeX{{%
    \normalfont B\kern-0.5em{\scshape i\kern-0.25em b}\kern-0.8em\TeX}}}

\setcopyright{acmlicensed}
\acmPrice{15.00}
\acmDOI{10.1145/3589250.3596144}
\acmYear{2023}
\copyrightyear{2023}
\acmSubmissionID{pldiws23soapmain-id3662-p}
\acmISBN{979-8-4007-0170-2/23/06}
\acmConference[SOAP '23]{Proceedings of the 12th ACM SIGPLAN International Workshop on the State Of the Art in Program Analysis}{June 17, 2023}{Orlando, FL, USA}
\acmBooktitle{Proceedings of the 12th ACM SIGPLAN International Workshop on the State Of the Art in Program Analysis (SOAP '23), June 17, 2023, Orlando, FL, USA}
\received{2023-03-10}
\received[accepted]{2023-04-21}

\begin{document}

\title{Combining E-Graphs with Abstract Interpretation}


\author{Samuel Coward}
\email{s.coward21@imperial.ac.uk}
\additionalaffiliation{%
  \institution{Numerical Hardware Group, Intel Corporation}
  \city{London}
}
\affiliation{%
  \institution{Electrical and Electronic Engineering \\
Imperial College London \\}
\country{UK}
}
\author{George A. Constantinides}
\email{g.constantinides@imperial.ac.uk}
\affiliation{%
  \institution{Electrical and Electronic Engineering \\
Imperial College London \\}
\country{UK}
}
\author{Theo Drane}
\email{theo.drane@intel.com}
\affiliation{%
  \institution{Numerical Hardware Group \\
Intel Corporation \\}
\country{USA}
}


\begin{abstract}
E-graphs are a data structure that compactly represents equivalent expressions. They are constructed via the repeated application of rewrite rules. Often in practical applications, {\em conditional} rewrite rules are crucial, but their application requires the detection -- at the time the e-graph is being built -- that a condition is valid in the domain of application. 
Detecting condition validity amounts to proving a property of the program. Abstract interpretation is a general method to learn such properties, traditionally used in static analysis tools. We demonstrate that abstract interpretation and e-graph analysis naturally reinforce each other through a tight integration
because (i) the e-graph clustering of equivalent expressions induces natural precision refinement of abstractions and (ii) precise
abstractions allow the application of deeper rewrite rules (and hence potentially even greater precision). 
We develop the theory behind this intuition and present an exemplar interval arithmetic implementation, which we apply to the FPBench suite. 
\end{abstract}

\ccsdesc[500]{Theory of computation~Abstraction}
\ccsdesc[500]{Theory of computation~Equational logic and rewriting}
\ccsdesc[500]{Mathematics of computing~Interval arithmetic}
\keywords{abstract interpretation, interval arithmetic, static analysis, e-graph}

\maketitle
\section{Introduction}
Equivalence graphs, commonly called e-graphs, provide a compact representation of equivalence classes (e-classes) of expressions, where the notion of equivalence is with respect to some concrete semantics \cite{Nelson1980TechniquesVerification}. The recent \texttt{egg} tool introduced {\em e-class analysis}, a technique to enable program analysis over an e-graph, attaching analysis data to each e-class~\cite{Willsey2021Egg:Saturation}. This paper formalises some of the concepts required to produce e-class analyses enabling e-graph growth via conditional rewrites. We show that partitioning expressions
into e-classes gives rise to a natural lattice-theoretic interpretation for abstract interpretation (AI),
resulting in the generation of precise abstractions. 
By interleaving conditional e-graph rewriting and program analysis the exploration power of the e-graph can be greatly enhanced. Figure~\ref{fig:cycle} visualizes this positive feedback loop.
Figure~\ref{fig:example_egraph} provides an
example, in which interval analyses of equivalent expressions are combined to produce tighter enclosing intervals.

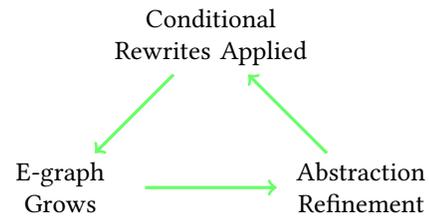
\begin{figure}
    \centering
    \begin{tikzpicture}

\def \n {3}
\def \radius {1.5cm}
\def \margin {14} 





\node[draw=none, text centered, text width = 3cm] at (0,2) (rewrite) {Conditional Rewrites Applied};
\node[draw=none, text centered, text width = 2cm] at (-2,0) (egraph) {E-graph Grows};
\node[draw=none, text centered, text width = 2cm] at (2,0) (abstraction) {Abstraction Refinement};

\draw[->, very thick, green!60] (rewrite) edge (egraph);
\draw[->, very thick, green!60] (egraph) edge (abstraction);
\draw[->, very thick, green!60] (abstraction) edge (rewrite);
\end{tikzpicture}
    \caption{The positive feedback loop between e-graph exploration and abstraction refinement. 
    }
    \label{fig:cycle}
\end{figure}

We develop the general theoretical underpinnings of AI on e-graphs, exploiting rewrites to
produce tight abstractions using a lattice-theoretic formalism. We also provide a sound interpretation of
cycles naturally arising in e-graphs, corresponding to extracting abstract fixpoint equations, and show that 
a known interval algorithm, the Krawczyk method, results as a special case.

The ideas presented here inspired the \texttt{egg} developers to implement interval analysis as part of their PLDI 2022 tutorial\footnote{https://github.com/egraphs-good/egg-tutorial-pldi-2022}. This paper extends and generalizes an abstract presented at the EGRAPHS workshop \cite{Coward2022AbstractE-Graphs} in the following ways:
\begin{itemize}
    \item formalization of AI with e-graphs in lattice theory, 
    \item relating fixpoints to e-graph cycles to automatically discover iterative abstract refinement methods, 
    \item techniques to capture relationships between variables,
    \item an interval arithmetic implementation with associated expression bounding results.
\end{itemize}

A short background overview is provided in \S \ref{sect:background}. In \S \ref{sect:theory} we present the theoretical application of AI to e-graphs and then demonstrate its viability using an interval arithmetic implementation in \S \ref{sect:impl}. Lastly, in \S \ref{sect:results} we present results. 

\begin{figure}
    \centering
    \includegraphics[scale=0.4]{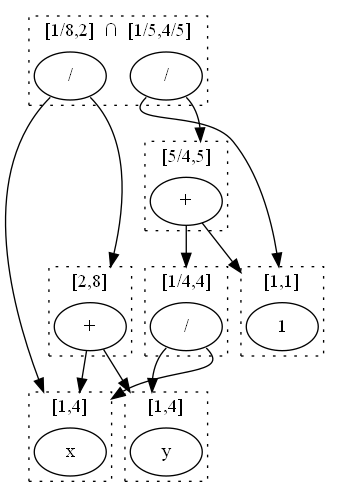}
    \caption{E-graph containing equivalent real arithmetic expressions $\frac{x}{x+y}$ and $\frac{1}{1+(y/x)}$ in the root e-class (at the top). Intervals are associated with each e-class. Input constraints are propagated upwards via an e-class analysis.}
    \label{fig:example_egraph}
\end{figure}
\section{Background}\label{sect:background}
The e-graph data structure is commonly found in theorem provers and solvers~\cite{DeMoura2008Z3:Solver,Detlefs2005Simplify:Checking}. It represents multiple expressions as a graph (e.g.~Figure \ref{fig:example_egraph}), where the nodes represent functions, grouped together into collections of e-classes. Edges connect nodes to e-classes, as a given sub-expression may be implemented using any of the nodes found in the child e-class. E-graphs are often combined with an optimisation technique called equality saturation~\cite{Willsey2021Egg:Saturation,Tate2009EqualityOptimization,Joshi2002Denali:Superoptimizer}, which deploys equivalence preserving transformations to monotonically grow the e-graph and discover alternative equivalent expressions. A recent resurgence of e-graph research~\cite{Willsey2021Egg:Saturation} has seen the technique applied to floating point numerical stability analysis~\cite{Panchekha2015AutomaticallyExpressions}, mapping programs onto hardware accelerators \cite{Smith2021PurePearl} and optimizing hardware designs~\cite{Coward2022AutomaticE-Graphs}. Some works make heavy use of conditional rewrites of the form $\phi \Rightarrow \ell \to r$, for example $x \neq 0 \Rightarrow x/x \to 1$, which require the e-graph construction algorithm to determine whether the rewrite is applicable in each case~\cite{Coward2022AutomaticE-Graphs,Willsey2021Egg:Saturation}.

AI has primarily been applied to static program analysis~\cite{Cousot1977AbstractFixpoints}, where computer programs are automatically analyzed without actually being executed. It uses the theory of abstraction to consider over-approximations of the program behaviour using alternative interpretations~\cite{Cousot1977AbstractFixpoints, Cousot2001AbstractChallenges}. 
Existing tools have incorporated term rewriting to refine their abstractions \cite{Daumas2010CertificationOperators}. 
In \S \ref{subsect:relational_domains} we will discuss relationships between variables allowing the framework to infer the consequences of conditional branches. Previous work on constraint programming has combined constraint awareness with interval analysis \cite{Talbot2019CombiningInterpretation}, whilst Granger identified that constraints are typically better evaluated in the abstract domain \cite{Granger1992ImprovingIterations}. 

For representing disjunctions of intervals, researchers have developed Linear Decision Diagrams (LDDs) and Range Decision Diagrams (RDDs) which are extensions of Binary Decision Diagrams (BDDs)~\cite{Gurfinkel2010Boxes:Boxes, Gange2021DisjunctiveAnalysis}. Previous work on abstract congruence closure \cite{Bachmair2003AbstractClosure,Bachmair2000AbstractSpecializations} has exploited e-graphs to combine different abstract domains \cite{Chang2005AbstractStructures}, but we are not aware of existing
work exploiting the tight interaction between AI and e-graph construction. 


\section{Theory} \label{sect:theory}
\subsection{Abstraction} \label{subsec:abstraction}
From a theoretical viewpoint, AI \cite{Cousot2021PrinciplesInterpretation} is concerned with relationships between lattices, defined via Galois connections.
\begin{definition}[Lattice]
A lattice is a partially ordered set (poset) $\langle L,\leq \rangle$, such that $\forall a,b \in L$ the least upper bound (join) $a \sqcup b$ and the greatest lower bound (meet) $a \sqcap b$ both exist. 
\end{definition}
\begin{definition}[Galois connection]
Given a poset $\langle \mathcal{C}, \sqsubseteq \rangle$, corresponding to the concrete domain, and a poset $\langle \mathcal{A},\preccurlyeq\rangle$, corresponding to the abstract domain, then a function pair $\alpha~\in~\mathcal{C}\rightarrow \mathcal{A}$, $\gamma \in \mathcal{A} \rightarrow \mathcal{C}$, defines a Galois connection iff
\[\forall P\in\mathcal{C}. \; \forall \overline{P}\in\mathcal{A}. \; \alpha(P) \preccurlyeq \overline{P} \Leftrightarrow P \sqsubseteq \gamma(\overline{P}) \textrm{,}\] 
\[\textrm{written }\langle \mathcal{C}, \sqsubseteq \rangle \overset{\alpha}{\underset{\gamma}\rightleftarrows} \langle \mathcal{A},\preccurlyeq \rangle.\]
\end{definition}
The pair $(\alpha,\gamma)$ define the abstraction and concretization respectively, allowing us to over-approximate (i.e.~abstract) concrete properties in $\mathcal{C}$ with abstract properties in $\mathcal{A}$.
\begin{definition}[Sound abstraction~\cite{Cousot2021PrinciplesInterpretation}]
$\overline{P} \in \mathcal{A}$ is a sound abstraction of a concrete property $P\in \mathcal{C}$ iff $P\sqsubseteq \gamma(\overline{P}).$
\end{definition}

Consider expressions evaluated over a domain $\mathcal{D}$. By imposing a canonical ordering on the variable set, we work within a defined subset $I\subseteq \mathcal{D}^n$, which encodes a precondition on the set of (input) variable values. 
Now consider a (concrete) semantics of expressions $\llbracket \cdot \rrbracket_\cdot \in \mathrm{Expr} \to I \to \mathcal{D}$, where $\mathrm{Expr}$ denotes the set of expressions, so $\llbracket e \rrbracket_\rho$ denotes the interpretation of expression $e$ under execution environment (assignment of variables to values) $\rho\in I$.
Let $\llbracket e \rrbracket =  \left \{\llbracket e \rrbracket_\rho \,|\, \rho \in I\right \}$.
The e-graph data structure encodes equivalence under concrete semantics, which we shall now
define precisely.

\begin{definition}[Congruence]
Two expressions $e_a$ and $e_b$ are congruent, $e_a \cong e_b$, iff $\llbracket e_a \rrbracket_\rho = \llbracket e_b \rrbracket_\rho$ for all $\rho \in I$.
\end{definition}
\begin{lemma}
If $e_a \cong e_b$ and $\overline{P}$ is a sound abstraction of $\llbracket e_a \rrbracket$, then $\overline{P}$ is a sound abstraction of $\llbracket e_b \rrbracket$.
\end{lemma}
\begin{proof}
by definition of congruence.
\end{proof}
This lemma implies that a sound abstraction of one expression in an e-class is a sound abstraction of all expressions in the e-class. Precision refinement relies on the following, which is a specialization of the more general result \cite{Cousot2011TheProcedures}. 
\begin{lemma}\label{lem:intersect_abstr}
For any two sound abstractions $\overline{P}_a$ and $\overline{P}_b$ of $P$, the meet $\overline{P}_a \sqcap \overline{P}_b$ is also a sound abstraction of $P$.
\end{lemma}
\begin{proof}
$P\sqsubseteq \gamma(\overline{P}_a)$ (sound abstraction) $\Rightarrow \alpha(P) \preccurlyeq \overline{P}_a$ (Galois connection) and similarly $P\sqsubseteq \gamma(\overline{P}_b) \Rightarrow \alpha(P) \preccurlyeq \overline{P}_b$. Therefore $\alpha(P) \preccurlyeq \overline{P}_a \sqcap \overline{P}_b$ (meet definition) and hence $P\sqsubseteq \gamma(\overline{P}_a \sqcap \overline{P}_b)$ (Galois connection).
\end{proof}

\subsection{Application to E-graphs}
Consider an e-graph. Let $S$ denote the set of e-classes, and $\mathcal{N}_s$ the set of nodes in the equivalence class $s\in S$. With each e-class, associate an abstraction $A \in \mathcal{A}$ and write $\mathcal{A}\llbracket s \rrbracket = A$.  Interpreting a $k$-arity node $n$ of function $f$ with children classes $s_1, ..., s_k$, using an arbitrary sound abstraction $\bar{f}$:
\begin{equation} \label{eqn:node_interp}
    \mathcal{A}\llbracket n \rrbracket = \bar{f}(\mathcal{A}\llbracket s_1 \rrbracket, ..., \mathcal{A}\llbracket s_k \rrbracket).
\end{equation}
$0$-arity nodes are either constants with exact abstractions in $\mathcal{A}$ or variables with user specified abstract constraints. 

For acyclic e-graphs, we propagate the known abstractions upwards using \eqref{eqn:node_interp}, taking the greatest lower bound (meet) across all nodes in the e-class.
\begin{equation}\label{eqn:class_prop}
    \mathcal{A}\llbracket s \rrbracket = \bigsqcap_{n\in \mathcal{N}_s} \mathcal{A}\llbracket n \rrbracket
\end{equation}
The propagation algorithm is described in Figure \ref{fig:property_prop}, where 
\[\texttt{make}(n) = \mathcal{A}\llbracket n \rrbracket \textrm{ and } \texttt{meet}(A_1, A_2)=A_1 \sqcap A_2.\]
These functions are analogous to those described for an e-class analysis~\cite{Willsey2021Egg:Saturation}, but replace their \texttt{join} with a \texttt{meet}. 

\begin{figure}
   \begin{lstlisting}[language=Python, basicstyle=\small]
    workqueue = egraph.classes().leaves()
    while !workqueue.is_empty()
      s = workqueue.dequeue()
      for n in s.nodes()
        skip_node = false
        for child_s in n.children()
          if child_s.uninitialized
            workqueue.enqueue(s)
            skip_node = true 
        if skip_node
            continue
        elif s.uninitialized
            s.data = make(n)
            s.uninitialized = false            
            workqueue.enqueue(s.parents())
        elif !(s.data<=meet(s.data,make(n)))
            s.data = meet(s.data,make(n))
            workqueue.enqueue(s.parents())
\end{lstlisting}
  \caption{
      Pseudocode for abstract property propagation in an e-graph.
    }
  \label{fig:property_prop}
\end{figure}

By lifting the abstract analysis from expressions to e-classes of expressions we construct a more precise analysis. In the abstract domain the notion of equivalence is different, $n_a, n_b \in \mathcal{N}_s \not \Rightarrow \mathcal{A}\llbracket n_a \rrbracket = \mathcal{A}\llbracket n_b \rrbracket$, which results in tighter abstractions since the meet corresponds to a more precise abstraction in $\mathcal{A}$.
In the algorithm in Figure~\ref{fig:property_prop}, by initializing the \texttt{workqueue} with only the modified e-classes after application of a rewrite, the abstract properties of the e-graph can be evaluated on the fly. For an acyclic e-graph, the propagation of each merge is worst case linear in the size of the e-graph. On-the-fly evaluation facilitates conditional rewrite application as more precise properties are discovered during construction. This is used in \S \ref{sect:impl}. 

A positive feedback loop is created by combining AI and e-graphs (Figure \ref{fig:cycle}). A larger space of equivalent expressions is explored as more rewrites can be proven to be valid at exploration time. In turn, expression abstractions are further refined by discovering more equivalent expressions, allowing even more valid rewrites, and the cycle continues. Additionally, several equivalent expressions may contribute to the tight final abstraction. An example is shown in \S \ref{sect:impl}. 

\subsection{Cyclic E-graphs and Fixpoints}\label{sect:cyclic}
Cyclic e-graphs arise when an expression is equivalent to a sub-expression of itself with respect to concrete semantics, for example $e \times 1\cong e$. Let $e \cong e'$ where $e$ appears as a subterm in $e'$. Treating
other subterms as absorbed into the function, let $f : \mathcal{D} \to \mathcal{D}$ be the interpretation of $e'$ as a
(concrete) function of $\llbracket e \rrbracket_\rho$, so that -- in particular -- $f(\llbracket e \rrbracket_\rho) = \llbracket e \rrbracket_\rho$ due to the congruence and hence $\llbracket e \rrbracket = \{f(\llbracket e \rrbracket_\rho) \,|\, \rho \in I\}$.
Abstracting $f$ via a sound abstraction $\bar{f}$, yields the corresponding abstract fixpoint equation $a = a \sqcap \bar{f}(a)$ where the meet operation arises from \eqref{eqn:class_prop}.

Now consider the function $\tilde{f}(a) = a \sqcap \bar{f}(a)$. The decreasing sequence defined by $a_{n+1} = \tilde{f}(a_n)$ corresponds to applying the abstract property propagation around a cycle in the e-graph, given an initial sound abstraction $a_0$ of $\llbracket e \rrbracket$.
\begin{lemma}
$\alpha( \llbracket e \rrbracket )$ is a fixpoint of $\tilde{f}$.  
\end{lemma}
\begin{proof}
\begin{align*}
    \alpha(\llbracket e \rrbracket) &= \alpha\left(\left\{f(\llbracket e \rrbracket_\rho) \,|\, \rho \in I\right \}\right)&\textrm{ (congruence)}\\
    &\preccurlyeq \bar{f}(\alpha(\llbracket e \rrbracket)) &\textrm{ (sound abstraction)}
\end{align*}
Hence $\tilde{f}( \alpha(\llbracket e \rrbracket) ) = \alpha(\llbracket e \rrbracket) \sqcap \bar{f}(\alpha(\llbracket e \rrbracket)) = \alpha(\llbracket e \rrbracket)$ (meet definition).
\end{proof}

\begin{lemma}
$a_{n}$ is a sound abstraction of $\llbracket e \rrbracket$ for all $n\in \mathbb{N}$.
\end{lemma}
\begin{proof}
By induction, $\llbracket e \rrbracket \sqsubseteq \gamma(a_0)$ and assume $\llbracket e \rrbracket \sqsubseteq \gamma(a_n)$. 
$\llbracket e \rrbracket = \left\{f(\llbracket e \rrbracket_\rho) \,|\, \rho \in I\right \} \sqsubseteq \gamma(\bar{f}(a_n))$ (sound abstraction of $f$). Hence $a_{n+1} = a_n \sqcap \bar{f}(a_n)$ is a sound abstraction of $\llbracket e \rrbracket$ (Lemma \ref{lem:intersect_abstr}) for all $n$.
\end{proof}

Collecting these results, for some fixpoint $a^*$ we have
\begin{equation}\label{eqn:fixpoint_chain}
\alpha(\llbracket e \rrbracket) \preccurlyeq a^{*} \preccurlyeq \ldots \preccurlyeq a_1 \preccurlyeq a_0.    
\end{equation}
Thus computing abstractions around the loop refines the abstraction and is guaranteed to terminate if the lattice $\langle \mathcal{A}, \preccurlyeq\rangle$ satisfies the descending chain condition, as any finite abstract domain will~\cite{Birkhoff1958VonTheory}. Note that the fixpoint $a^*$ is neither greatest, least nor unique. This can be seen because the top and bottom elements of the lattice are both also fixedpoints of $\tilde{f}$, but from \eqref{eqn:fixpoint_chain} the computed fixedpoint is between these two other fixedpoints. The algorithm in Figure \ref{fig:property_prop} will correctly apply abstract property propagation around loops, terminating if the sequence $a_n$ converges in a finite number of steps. Of course for abstract domains with infinite descending chains standard techniques such as narrowing apply \cite{Cousot1992ComparingInterpretation}.



\subsection{Relational Domains} \label{subsect:relational_domains}

For simplicity of exposition, we focus on non-relational domains in the discussion above. However, it is important to note that the combination of non-relational domains with the {\em relational} information provided by rewrite rules provides a stronger analysis than classical non-relational domains. By introducing an additional $\texttt{ASSUME}$ node that captures the domain refinement implied by design constraints, relationships between variables can be captured, enabling further rewrites. We illustrate this point via an example.
\begin{equation} \label{eqn:relational_start}
x==y \, ? \, x+y : 0.    
\end{equation}
A human quickly realises the relationship between $x$ and $y$ in the true branch and can replace $x+y$ by $2\times y$. In the e-graph framework, the expression is automatically rewritten to identify `new' terms which simplify the conditional reasoning. 
As the e-graph grows \eqref{eqn:relational_start} will be found to be equivalent to:
\begin{equation} \label{eqn:relational_int}
z==0 \, ? \, z+2\times y : 0,\quad \textrm{where }z=x-y.    
\end{equation}
The e-graph identifies a new variable, $z$, which must be a constant, zero, on the true branch. The $\texttt{ASSUME}$ node is introduced by rewriting the ternary operation. 
\begin{equation} \label{eqn:assume_int}
z==0 \, ? \, \texttt{ASSUME}(z+2\times y,\,z==0) : \texttt{ASSUME}(0,\,z\neq 0)    
\end{equation}
The $z==0$ constraint is automatically pushed down the expression tree. At the point at which the propagated condition reaches $z$ it is automatically rewritten to constant zero. The e-graph continues to grow, until it discovers the simplified equivalent expression, where the $\texttt{ASSUME}$ nodes are dropped.
\begin{equation} \label{eqn:relational_final}
z==0 \, ? \, 2\times y : 0,\quad \textrm{where }z=x-y.    
\end{equation}

We combined an existing hardware optimization e-graph tool \cite{Coward2022AutomaticE-Graphs} with the program analysis techniques presented here, where we introduce $\texttt{ASSUME}$ nodes in more detail \cite{Coward2023AutomatingE-Graphs}. In this work, we implement rewriting over real valued expressions but in~\cite{Coward2023AutomatingE-Graphs} we apply the techniques presented here to fixed width integer valued expressions. 


\section{Implementation}\label{sect:impl}

To demonstrate the theory described above, we implement interval arithmetic (IA)~\cite{Moore2009IntroductionAnalysis} for real valued expressions using the extensible \texttt{egg} library, as an e-class analysis \cite{Willsey2021Egg:Saturation}. We consider a concrete domain corresponding to sets of extended real numbers, i.e. $\mathcal{C} = \mathcal{P}(\mathbb{R} \cup \{-\infty, +\infty\})$ where $\mathcal{P}$ denotes the power set. We associate each expression with a \texttt{binary64}~\cite{2008IEEEArithmetic} valued interval (a finite abstract domain),
\[\mathcal{A} = \left\{[a,b] \; | \; a \leq b, a,b \in \texttt{binary64}\setminus \{\texttt{NaN}\} \right\} \cup \{\emptyset\}.\] 
In this setting the abstraction and concretization functions are as follows (infima/suprema always exist in this setting):
\begin{align}
\alpha(X) &= \left[\textrm{round\_down}(\inf X), \textrm{round\_up}(\sup X)\right] \\ 
\gamma([a,b]) & = [a, b]\\
\alpha(\emptyset) &= \emptyset, \gamma(\emptyset) = \emptyset
\end{align}

For this work we supported the following set of operators, +, -, $\times$, /, $\sqrt{}$, \texttt{pow}, \texttt{exp} and \texttt{ln}. 
To ensure correctness, we use `outwardly rounded IA' which conservatively rounds upper bounds towards $+ \infty$ (round\_up) and lower bounds towards $- \infty$ (round\_down) \cite{Moore2009IntroductionAnalysis,Kulisch1981ComputerPractice}.
For the elementary functions, $\sqrt{},\texttt{exp}$ and $\texttt{ln}$, we use default library implementations but are unable to control the rounding mode, so conservatively add or subtract one unit in the last place for upper and lower bounds respectively. 
Provided \texttt{NaN}s do not appear in the input expression evaluation they are not generated by the e-graph exploration. Furthermore abstract intervals containing $-0$ shall be mapped by $\gamma$ to sets containing $0\in \mathcal{C}$.

In this case we use an abstraction of a given function $f$, $\bar{f}~=~\alpha~\circ~f~\circ~\gamma$.
For e-class $s$ under this interpretation, \eqref{eqn:class_prop} uses the intersection operation, the meet operation of the lattice of intervals.
\begin{equation}\label{eqn:int_class_update}
    \mathcal{A}\llbracket s \rrbracket = \bigcap_{n\in \mathcal{N}_s} \llbracket n \rrbracket
\end{equation}
This relationship generates monotonically narrowing interval abstractions. 0-arity nodes represent either constants associated with degenerate intervals or variables taking user defined interval constraints.


The classical problem of interval arithmetic is the so-called `dependency problem', arising because the domain does not capture correlations between multiple occurrences of a single variable. Consider $x\in [0,1]$, under classical IA:
\begin{equation} \label{eqn:ia_limits}
\mathcal{A}\llbracket x-x \rrbracket = [0,1]-[0,1] = [0-1, 1-0] = [-1,1].   
\end{equation}
The e-graph framework discovers, via term rewriting, $x-x\cong 0$ and by \eqref{eqn:int_class_update} the expression is now correctly abstracted by the (much tighter) degenerate interval $[0,0]$.

\begin{table}
    \centering
    \caption{Additional IA optimization rewrites.}
    \begin{tabular}{|c|c|c|}
    \hline
    Class & Rewrite & Condition \\
    \hline 
    \mycc Factor       & \mycc $ab \pm a c \rightarrow a (b\pm c)$ & \mycc True\\
    Binom. & $1/(1-a) \rightarrow 1 + a / (1 - a)$ & $0\not \in \llbracket 1-a \rrbracket$\\
    \mycc Frac   & \mycc $b/c \pm a \rightarrow (b \pm a c)/c$& \mycc $0 \not \in \llbracket c\rrbracket$\\
    \multirow{2}{4em}{Div. \cite{Moore2009IntroductionAnalysis}}  & $ a/b \rightarrow 1/(b/a)$ & $0\not \in \llbracket a \rrbracket \cup \llbracket b \rrbracket$\\
      & $a/b \rightarrow 1 +(a-b)/b$ & $0 \not \in \llbracket b \rrbracket$\\
    \mycc Poly & \mycc $a^2 - 1 \rightarrow (a - 1)(a+1)$ & \mycc True \\
    Elem. & $\ln(e^a) \rightarrow a$ & True\\
    \hline    
    \end{tabular}
    \label{tab:rewrite_table}
\end{table}

We use a set of 39 rewrites, defining equivalences of real valued expressions. The basic arithmetic rewrites are commutativity, associativity, distributivity, cancellation and idempotent operation reduction across addition, subtraction, multiplication and division. Conversion rewrites describe the natural equivalence between the power function and multiplication/division. Table \ref{tab:rewrite_table} contains the remaining rewrites. 

Conditional rewrites, e.g. ``Div.'', are only valid on a subset of the input domain. Via an IA the e-graph can prove validity of such rules. In \eqref{eqn:cond_apply} IA can confirm that $0\not \in \llbracket x+y \rrbracket$, in order to remove multiple occurrences of variables resulting in expression bound improvements. 
\begin{equation}\label{eqn:cond_apply}
    \frac{x+y}{x+y+1}\rightarrow ... \rightarrow \frac{1}{1+\frac{1}{x+y}}
\end{equation}




Multiple expressions in an e-class can independently contribute to a tight abstraction. Consider the following equivalent expressions for variables $x \in [0,1]$ and $y \in [1,2]$.
\begin{align}
    & 1 - \frac{2y}{x+y} &\in \left[-3,\frac{1}{3}\right] \label{eqn:rw_a}  \\
    &\cong  \frac{x-y}{x+y}    &\in [-2,0] \label{eqn:rw_b}\\
    &\cong  \frac{2x}{x+y} - 1 &\in [-1,1] \label{eqn:rw_c}.
\end{align}
All three reside in the same e-class within an e-graph with associated interval $[-3,\frac{1}{3}] \cap [-2,0] \cap [-1,1] = [-1,0]$. Thus, given the first expression~\eqref{eqn:rw_a}, the e-graph generates a tight interval enclosure using two distinct equivalent expressions for the upper \eqref{eqn:rw_b} and lower \eqref{eqn:rw_c} bounds.

\section{Results}\label{sect:results}
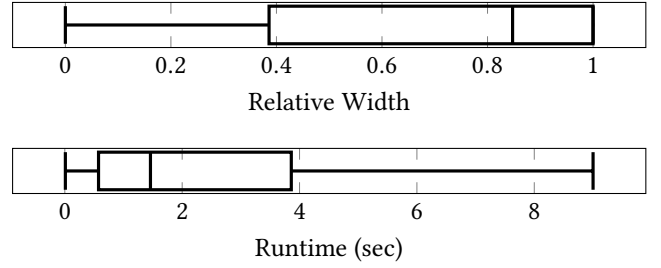
\begin{figure}
    \centering
    \begin{tikzpicture}
\begin{axis}[
y=0.5cm,
ytick=\empty,
xlabel=Relative Width,
]
\addplot+[
  boxplot prepared={
    lower whisker=0,
    lower quartile=0.386,
    median=0.8475,
    upper quartile=1,
    upper whisker=1,
    box extend=1,  
    whisker extend=1, 
    every box/.style={very thick,draw=black},
    every whisker/.style={black,very thick},
    every median/.style={black, very thick},
  },
]
table[row sep=\\,y index=0] {
data\\ 
};
\end{axis}
\end{tikzpicture}

\vspace{1em}

\begin{tikzpicture}
\begin{axis}[
y=0.5cm,
ytick=\empty,
xlabel=Runtime (sec),
]
\addplot+[
  boxplot prepared={
    lower whisker=0.0093991,
    lower quartile=0.5746177,
    median=1.4594177,
    upper quartile=3.8594865,
    upper whisker=9.001914,
    box extend=1,  
    whisker extend=1, 
    every box/.style={very thick,draw=black},
    every whisker/.style={black,very thick},
    every median/.style={black, very thick},
  },
]
table[row sep=\\,y index=0] {
data\\ 
};
\end{axis}
\end{tikzpicture}
    \caption{Relative interval width (optimized width/naive width) and runtime boxplots to demonstrate the distribution of results on the FPBench suite. 
    }
    \label{fig:boxplots}
\end{figure}

Based on the theory introduced above we describe a real valued expression bounding tool in Rust using the \texttt{egg} library. All test cases were run on an Intel i7-10610U CPU.

\subsection{Benchmarks}
We evaluate the implementation using 40 benchmarks from the FPTaylor \cite{Solovyev2018RigorousExpansions} supported subset of the FPBench benchmark suite \cite{Damouche2017TowardAnalysis}. We allow four iterations of e-graph rewriting, as further rewriting iterations do not yield significant improvements in interval width on these modest benchmarks. 

Across these benchmarks, the inclusion of IA and domain specific rewrites increased the number of e-graph nodes by 4\% on average but by up to 84\%. This demonstrates the additional rewrites that have been applied as a result of combining e-graphs and AI. Figure \ref{fig:boxplots} summarises the distribution of the interval width improvement and runtime across the benchmarks. The relative width boxplot shows that, on average, e-graph rewriting reduced the interval to 85\% of the width of the naive interval approximation. In a number of cases, the interval obtained from e-graph rewriting was reduced by almost 100\%. The runtime boxplot shows that runtimes remained small. There is little correlation between the runtime and bound improvement. The overhead of incorporating IA into the e-graph increased runtimes by less than 1\% on average.


\subsection{Iterative Method Discovery}
In \S \ref{sect:cyclic} we discussed how cyclic e-graphs can discover iterative refinements. The Krawczyk method~\cite{Moore2009IntroductionAnalysis} is a known algorithm to generate increasingly precise element-wise interval enclosures of solutions of linear systems of equations $Ax = b$, where $A$ is an $n$-by-$n$ matrix and $b$ is an $n$-dimensional vector.
Letting $X^0$ be an initial interval enclosure of the solutions, the Krawcyzk method uses an update formula of the form,
\begin{equation*}
X^{k+1} = \left(Yb + (I-YA)X^k\right) \cap X^k,\textrm{ where } Y=\textrm{mid}(A)^{-1}.
\end{equation*}
mid(A) is the element-wise interval midpoint of the matrix. This sequence, via interval extension and intersection, corresponds to a sequence of tightening bounds on the solution $x$, which converges provided the matrix norm $||I-YA||<1$.

We consider a specific instance of this problem,
\begin{equation}
    \begin{pmatrix}
    1 & y\\
    y & 1
    \end{pmatrix} 
    \begin{pmatrix}
    x_1\\
    x_2
    \end{pmatrix} =
    \begin{pmatrix}
    b_1\\
    b_2
    \end{pmatrix}
    \textrm{, where } y\in \left[-\frac{1}{2},\frac{1}{2}\right].
\end{equation}
\begin{equation} \label{eqn:kraw_x1}
    X_1^{k+1} = \left (b_1 - yX_2^k \right) \cap X_1^k,\;
    X_2^{k+1} = \left (b_2 - yX_1^k \right) \cap X_2^k.
\end{equation}

A naive solution in the concrete domain, $x = A^{-1} b$, yields,
\begin{equation}
x_1 = \frac{1}{1-y^2}(b_1 - b_2y), \; x_2 = \frac{1}{1-y^2}(b_2 - b_1y).
\end{equation}
Initialising the e-graph with these expressions, the solution for $x_1$ can be automatically rewritten such that \eqref{eqn:kraw_x1} arises in the abstract domain. The ``Binom.'' rewrite from Table \ref{tab:rewrite_table} introduces a loop into the e-graph, which when combined with distributivity rules and ``Factor'' from Table \ref{tab:rewrite_table} yields,

\begin{equation}
    x_1 = b_1 - y\left(b_2 + (b_2 y^2 - b_1 y)\frac{1}{1-y^2}\right).
\end{equation}
After applying, ``Frac'', and cancelling, the e-graph contains
\begin{equation}
    x_1 = b_1 - y(b_2 - b_1 y)\frac{1}{1-y^2} = b_1 - y x_2. 
\end{equation}
When a cycle is introduced, the IA update procedure will continue to iteratively evaluate the loop, taking the intersection with the previous iteration as described in \S \ref{sect:cyclic}.



\section{Conclusion}\label{sect:conclusion}
We present a combination of abstract interpretation and e-graphs, demonstrating the natural interpretation
of e-class partitions as meet operators in a lattice, resulting in precise abstractions. Of key importance
is the positive feedback loop between e-graph exploration and abstraction refinement, as the precision 
then allows the application of conditional rewrite rules, which can be applied in many domains and may
further improve abstraction precision. An exemplar interval arithmetic implementation has demonstrated the value of this idea, including
automated discovery of a known algorithm for iterative refinement. 
Furthermore, an existing hardware optimization e-graph application has already benefitted from incorporating these analysis techniques. We believe that other e-graph applications could also incorporate AI to extend their capabilities for relatively low overhead. 
Future work will explore additional abstract domains and their incorporation into e-graph optimization tools.

\bibliographystyle{ACM-Reference-Format}
\bibliography{references}

\end{document}